\documentclass[conference, letterpaper]{IEEEtran}
\pdfoutput=1
\IEEEoverridecommandlockouts
\usepackage{cite}
\usepackage{amssymb,amsfonts,amsthm}
\usepackage[cmex10]{amsmath}
\usepackage{algorithmic}
\usepackage{graphicx}
\usepackage{textcomp}
\usepackage{tikz}\usetikzlibrary{positioning,arrows}
\usepackage{xcolor}
\usepackage[utf8]{inputenc} 
\usepackage[T1]{fontenc}
\usepackage{empheq}
\usepackage{url}
\usepackage{ifthen}
\usepackage{subcaption}
\usepackage[section]{placeins}
\usepackage{calrsfs}
\usepackage[colorlinks,allcolors=black]{hyperref}

\newtheorem{theorem}{Theorem}
\newtheorem{corollary}{Corollary}
\newtheorem{proposition}{Proposition}

\theoremstyle{definition}
\newtheorem{definition}{Definition}
\newtheorem{example}{Example}
\newtheorem{notation}{Notation}
\newtheorem{remark}{Remark}

\newcommand{\numberset}{\mathbb}

\newcommand{\Z}{\numberset{Z}}

\newcommand{\R}{\numberset{R}}

\newcommand{\F}{\numberset{F}}

\newcommand{\mP}{\mathcal{P}}
\newcommand{\mT}{\mathcal{T}}

\newcommand{\floor}[1]{{\left\lfloor{#1}\right\rfloor}}

\newcommand{\PG}{\textnormal{PG}}

\newcommand{\col}{\textnormal{col}}
\newcommand{\conv}{\textnormal{conv}}

\newcommand{\mR}{\mathcal{R}}

\newcommand\qbin[3]{\left[\begin{matrix} #1 , #2 \end{matrix} \right]_{#3}}

\newcommand{\ceil}[1]{{\left\lceil{#1}\right\rceil}}

\def\BibTeX{{\rm B\kern-.05em{\sc i\kern-.025em b}\kern-.08em
T\kern-.1667em\lower.7ex\hbox{E}\kern-.125emX}}

\title{On the Parameters of Codes for Data Access\thanks{A.B.K. is supported by the Dutch Research Council through grant VI.Vidi.203.045. A. R. is supported
by OCENW.KLEIN.539, VI.Vidi.203.045, and by the
Academy of Arts and Sciences of the Netherlands. E.~Soljanin's work was in part supported by NSF Award CIF-2122400.}}

\author{%
   \IEEEauthorblockN{\textbf{Altan B. K{\i}l{\i}\c{c}}\IEEEauthorrefmark{1}, \textbf{Alberto Ravagnani}\IEEEauthorrefmark{1}, \textbf{Emina Soljanin}\IEEEauthorrefmark{2}}
    \IEEEauthorblockA{\IEEEauthorrefmark{1}%
    Department of Mathematics and Computer Science, Eindhoven University of Technology, the Netherlands}
    \IEEEauthorblockA{\IEEEauthorrefmark{2}%
    Department of Electrical \& Computer Engineering, %\\
    Rutgers University, U.S.A.}
    \texttt{\{a.b.kilic, a.ravagnani\}@tue.nl, emina.soljanin@rutgers.edu}
}

\begin{document}

\maketitle

%%%%%%%%%%%%%%%%%%%%%%%%%%%%%%%%%%%%%%%%%%%%%%%%%%%%%%%%%%%%%%%%%%%%%%%%%%%
%%%%%%%%%%%%%%%%%%%%%%%%%%%%%%%%%%%%%%%%%%%%%%%%%%%%%%%%%%%%%%%%%%%%%%%%%%%

\begin{abstract}
This paper studies two crucial problems in the context of coded distributed storage systems directly related to their performance: 1) for a fixed alphabet size, determine the minimum number of servers the system must have 
for its service rate region
to contain a prescribed set of points; 2)
for a given number of servers,
determine 
the minimum alphabet size for which the service rate region of the system contains a prescribed set of points. 
The paper establishes rigorous upper and lower bounds, as well as code constructions
based on techniques from coding theory, optimization, and projective geometry. 
\end{abstract}

%%%%%%%%%%%%%%%%%%%%%%%%%%%%%%%%%%%%%%%%%%%%%%%%%%%%%%%%%%%%%%%%%%%%%%%%%%%
%%%%%%%%%%%%%%%%%%%%%%%%%%%%%%%%%%%%%%%%%%%%%%%%%%%%%%%%%%%%%%%%%%%%%%%%%%%

\section{Introduction}
\label{sec:intro}

Storage layers provide data access services for executing applications. Thus, a computing system’s overall performance depends on its underlying storage system’s data access performance. Modern storage systems replicate data objects across multiple servers. An object’s replication degree corresponds to its expected demand. However, the access request volume and data object popularity fluctuate in practice. Redundancy schemes that combine erasure coding with replication can support such scenarios better than replication alone. 

Storage servers can handle access requests up to a specific maximal service rate.
The service rate region of a redundant storage scheme is a recently introduced performance metric notion \cite{aktacs2017service}; see also \cite{aktacs2021service} and references therein. It is defined as a set of all data access request rates that the system implementing the scheme can support. 

Two main problem directions are associated with distributed storage access: 1) For a given (implemented) redundancy scheme, we ask what its service rate region is. 2) For a given (desired) service rate region, we ask which redundancy scheme has the service rate region that includes the desired one with some optimal cost or some required properties. Many other related problems are briefly outlined in \cite[Sec.~VIII]{aktacs2021service}, most notably performance analysis of storage schemes.  

The service rate region is a new concept and the subject of several recent papers \cite{aktacs2017service,aktacs2021service,kazemi2020combinatorial,kazemi2020geometric,alfarano2022dual}.
This work has answered some questions in the above directions, primarily for selected binary codes or systems storing two data objects. 
The main difference between this collection of papers and nearly all recent work on coded distributed storage is that it primarily addresses the external uncertainty in the storage systems (download request fluctuations) rather than the internal uncertainty (e.g., straggling) in operations of the system itself (see, e.g.,
\cite{joshi2012coding, gardner2015reducing, shah2016when, rashmi2016ec,lee2017speeding, dutta2016short, mallick2020rateless}).
A related line of work concerning external uncertainty considers systems with uncertainty in the mode and level of access to the system \cite{alloc:PengNS21, alloc:PengS18}.

This paper addresses a problem within the second main direction mentioned above: designing codes for a given service rate region. It focuses on scenarios when a set of the points in the region is provided and asks the following two questions: 
\begin{enumerate}
    \item How many servers do we need if the field size is fixed?
    \item What is the minimum field size for a code over a fixed number of servers (i.e., code length)? \label{e2}
\end{enumerate}
Both questions are essential in practice. The first question concerns systems with computational complexity (field size) limits, and we want to minimize the number of servers. The second question is relevant when the number of servers is limited, and we want to minimize the field size to reduce the computational complexity. The first problem was studied in \cite{kazemi2021efficient}. The second problem has not been studied before.

The rest of the paper is organized as follows: Sec.~\ref{sec:srr} formulates the problem. Sec.~\ref{sec:methods}, presents two general approaches to attack the problems of the paper. Sec.~\ref{sec:bound} is devoted to bounds and existence results. Sec.~\ref{sec:interplay} shows some applications of the paper's results. Because of space constraints, most proofs are omitted and will appear in an extended version of this work.

%%%%%%%%%%%%%%%%%%%%%%%%%%%%%%%%%%%%%%%%%%%%%%%%%%%%%%%%%%%%%%%%%%%%%%%%%%%
%%%%%%%%%%%%%%%%%%%%%%%%%%%%%%%%%%%%%%%%%%%%%%%%%%%%%%%%%%%%%%%%%%%%%%%%%%%

\section{Service Rate Region and Problem Statement}\label{sec:srr}
In this section, we establish the notation, define the service rate region of distributed data storage systems, and present the parameters that are the subject of the two main problems 
addressed in this paper. For a standard reference on coding theory, we refer to \cite{macwilliams1977theory}.

\begin{notation}
Throughout the paper, $\mP$ denotes the set of prime powers, $\F_q$ denotes a finite field of size $q$ with $q \in \mP$. Unless otherwise stated, all vectors in the paper are column vectors, and $e_i \in \F_q^k$ is the $i$-th standard basis vector for $i \in \{1,\ldots,k\}$. 
\end{notation}

We consider a distributed storage system with $n \in \Z_{\ge k}$ identical servers where $k \in \Z_{\ge 2}$ distinct data objects (elements of~$\F_q$) are stored on the servers. Each server stores an $\F_q$-linear combination of the $k$ objects. Therefore, the system can be specified by a matrix $G \in \F_q^{k \times n}$, called the \textbf{generator matrix} of the system. Each column of~$G$ corresponds to a server. More precisely, if $u = (u_1,\ldots,u_k)$ is the tuple of objects to be stored, then the $\nu$-th column of $G$ stores the $\nu$-th coordinate of $u\cdot G$ for $\nu \in \{1,\ldots,n\}$. If a column of $G$ is a non-zero scalar multiple of one of the standard basis vectors, it is called a \textbf{systematic server}. Otherwise, it is called a \textbf{coded server}.

\begin{notation}
    In the sequel, $G$ denotes a $k \times n$ matrix over $\F_q$ of rank $k$. We assume that $G$ has no zero column. The list of columns of $G$ is denoted by $\col(G)$, and the $\nu$-th column of~$G$ is denoted by $G^\nu$. 
\end{notation}

In our model, each server can simultaneously process the data access requests with a cumulative serving rate of at most one.
A \textbf{demand vector} is a $k$-tuple of nonnegative real numbers 
$(\lambda_1, \ldots, \lambda_k)$. The assumption that $G$ has no zero column makes sense since a server storing the zero linear combination is useless. The restriction on the capacity of the servers (each server has capacity 1) implies that
the system cannot ``support'' any demand vector. When a server 
reaches its capacity, other servers need to be contacted to serve users.
The next definition formalizes these concepts.

\begin{definition}
\label{def:RandSRR}
For each $i \in \{1,\ldots,k\}$ we construct sets $\mR_i(G)$ as follows. A set $R \subseteq \{1,\ldots,n\}$ is in $\mR_i(G)$ if:
\begin{itemize}
    \item $e_i \in \langle G^\nu \mid \nu \in R\rangle_{\F_q}$, and
    \item there is no $R' \subsetneq R$ with $R' \in \mR_i(G)$.
\end{itemize}
The elements of $\mR_i(G)$ 
are called the (\textbf{minimal}) \textbf{recovery sets} for the $i$-th object.
A demand vector $(\lambda_1,\ldots,\lambda_k)$ is \textbf{supported} by the system (or supported by $G$) if there exists real numbers  
$$(\lambda_{i,R} \in \R_{\ge 0} \mid i \in \{1,\ldots,k\}, \, R \in \mR_i(G))$$ with the following properties:
\begin{empheq}[left = \empheqlbrace]{alignat=3}
        \sum_{R \in \mR_i}\lambda_{i,R} &= \lambda_i &&\textnormal{ for } i \in \{1,\ldots,k\}, \label{srr1}\\
        \sum_{i=1}^k \sum_{\substack{R \in \mR_i \\ \nu \in R}} \lambda_{i,R} &\leq 1  \ &&\textnormal{ for } \nu \in \{1,\ldots,n\}. \label{srr2} 
    \end{empheq}
Lastly, the \textbf{service rate region} of the system generated by $G$ is defined as 
$$\Lambda(G):= \{\lambda \in \R^k \mid \lambda \mbox{ is supported by } G\} \subseteq \R^k.$$
\end{definition}
\noindent
The constraints in \eqref{srr1} ensure that the demand vector is served, and those in \eqref{srr2} guarantee that the servers are not overloaded. 

\begin{example}
\label{ex1}
Let $$ G= \begin{pmatrix}
1 & 0 & 0 & 1   \\
0 & 1 & 0 & 1   \\
0 & 0 & 1 & 1 
\end{pmatrix}\in\F_2^{3\times 4}.$$ 
We have $\mR_i(G) = \{\{i\},\{1,2,3,4\}\setminus \{i\}\}$ for $i \in \{1,2,3\}$. The service rate region $\Lambda(G)$ is
depicted in Figure~\ref{fig:ex1}. As an example, we have $(1.4,0.6,0.6) \in \Lambda(G)$ by letting 
\begin{align*}
    \lambda_{1,R}=\begin{cases}
        1 & \textnormal{ if } R=\{1\},\\
        0.4 & \textnormal{ if } R=\{2,3,4\},\\
            \end{cases} \
     \lambda_{i,R}=\begin{cases}
        0.6 & \textnormal{ if } R=\{i\},\\
        0 & \textnormal{ otherwise,}
    \end{cases}  
\end{align*}
for $i \in \{2,3\}$ in Definition \ref{def:RandSRR}. 
\end{example}

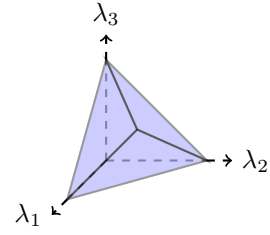
\begin{figure}[h!]
    \centering
    \begin{tikzpicture}[thick,scale=0.67]
 \coordinate (A1) at (0,0,0);
 \coordinate (A2) at (2,0,0);
 \coordinate (A3) at (0,2,0);
 \coordinate (A4) at (0,0,2);
 \coordinate (A5) at (1,1,1);

 \begin{scope}[thick,dashed,opacity=0.5]
 \draw (A1) -- (A2);
 \draw (A1) -- (A3);
 \draw (A1) -- (A4);

 \end{scope}

 \draw[fill=blue!50,opacity=0.4] (A5) -- (A2) -- (A3) -- cycle;
 \draw[fill=blue!50,opacity=0.4] (A5) -- (A2) -- (A4)  -- cycle;
 \draw[fill=blue!50,opacity=0.4] (A5) -- (A4) -- (A3) -- cycle;
 %\draw[fill=blue!50,opacity=0.4] (A8) -- (A5) -- (A6) -- (A7) -- (A4) -- cycle;

 \draw[->, dashed, thick,black] (2,0,0)--(2.5,0,0) node[right]{$\lambda_2$};
 \draw[->, dashed, thick,black] (0,0,2)--(0,0,2.8) node[left]{$\lambda_1$};
 \draw[->, dashed, thick,black] (0,2,0)--(0,2.5,0) node[above]{$\lambda_3$};
% \draw[thick,black] (A3)--(A8)--(A4)--(A7)--(A2)--(A3);
 \end{tikzpicture}
        \caption{The service rate region for Example~\ref{ex1}.}
    \label{fig:ex1}
\end{figure}

Note that the service rate region in Example \ref{ex1} is a polytope. Recall that a polytope $P \subseteq \R^k$ is called \textbf{down-monotone} if~$0 \le x \le y$ and $y \in P$ imply $x \in P$. The service rate region is always a convex, down-monotone polytope.
Therefore, we will call it the 
\textbf{service rate region polytope} of the matrix $G$, see \cite{service2023polytope}.
Recall that any convex polytope is the convex hull of its vertices  (see \cite{ziegler2012lectures} for more on polytopes). Therefore, we have the following list of observations.

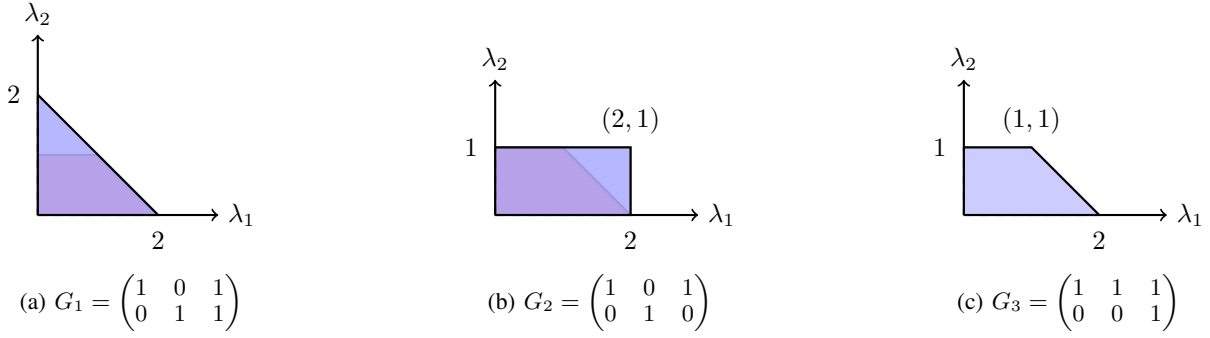
\begin{figure*}
    \centering
\subcaptionbox{$G_1 = \begin{pmatrix}
1 & 0 & 1   \\
0 & 1 & 1 
\end{pmatrix}$}
[.3\linewidth]{\begin{tikzpicture}[thick,scale=0.8]
 \coordinate (A1) at (0,0);
 \coordinate (A2) at (2,0);
 \coordinate (A4) at (1,1);
 \coordinate (A3) at (0,2);
 \coordinate (A5) at (0,1);
\node at (A3) [below = 1mm of A2] {$2$};
\node at (A2) [left = 1mm of A3] {$2$};
 \begin{scope}[thick,dashed,,opacity=0.6]
 \draw (A1) -- (A2);
 \draw (A1) -- (A3);

 \end{scope}

 \draw[fill=red!70,opacity=0.2] (A1) -- (A2) -- (A4) -- (A5)-- cycle;
 \draw[fill=blue!70,opacity=0.4] (A1) -- (A2) -- (A3)-- cycle;
 \draw[->, thick,black] (0,0)--(3,0) node[right]{$\lambda_1$};
\draw[->, thick,black] (0,0)--(0,3) node[above]{$\lambda_2$};
\draw[thick,black] (0,2)--(2,0);
 \end{tikzpicture}}
 \hspace{.03\textwidth}
\subcaptionbox{$G_2 = \begin{pmatrix}
1 & 0 & 1   \\
0 & 1 & 0 
\end{pmatrix}$}
[.3\linewidth]{\begin{tikzpicture}[thick,scale=0.9]
\coordinate (A1) at (0,0);
 \coordinate (A2) at (2,0);
 \coordinate (A3) at (0,1);
 \coordinate (A4) at (2,1);
 \coordinate (A5) at (1,1);
\node at (A3) [left = 1mm of A3] {$1$};
\node at (A2) [below = 1mm of A2] {$2$};
\node at (A4) [right = 0.5mm, above = 0.5mm of A4] {$(2,1)$};
 \begin{scope}[thick,dashed,,opacity=0.6]
 \draw (A1) -- (A2);
 \draw (A1) -- (A3);

 \end{scope}
 
 \draw[fill=red!70,opacity=0.2] (A1) -- (A2) -- (A5)-- (A3)--cycle;
 \draw[fill=blue!70,opacity=0.4] (A1) -- (A2) -- (A4)-- (A3) -- cycle;
 \draw[->, thick,black] (0,0)--(3,0) node[right]{$\lambda_1$};
\draw[->, thick,black] (0,0)--(0,2) node[above]{$\lambda_2$};
\draw[thick,black] (0,1)--(2,1)--(2,0);
\end{tikzpicture}}
 \hspace{.03\textwidth}
\subcaptionbox{$G_3 = \begin{pmatrix}
1 & 1 & 1   \\
0 & 0 & 1 
\end{pmatrix}$}
[.3\linewidth]{\begin{tikzpicture}[thick,scale=0.9]
\coordinate (A1) at (0,0);
 \coordinate (A2) at (2,0);
 \coordinate (A3) at (0,1);
 \coordinate (A4) at (1,1);
\node at (A3) [left = 1mm of A3] {$1$};
\node at (A2) [below = 1mm of A2] {$2$};
\node at (A4) [right = 0.5mm, above = 0.5mm of A4] {$(1,1)$};
 \begin{scope}[thick,dashed,,opacity=0.6]
 \draw (A1) -- (A2);
 \draw (A1) -- (A3);

 \end{scope}
 
 %\draw[fill=red!70,opacity=0.2] (A1) -- (A2) -- (A4)-- (A3) --  cycle;
 \draw[fill=blue!50,opacity=0.4] (A1) -- (A2) -- (A4)-- (A3) -- cycle;
 \draw[->, thick,black] (0,0)--(3,0) node[right]{$\lambda_1$};
\draw[->, thick,black] (0,0)--(0,2) node[above]{$\lambda_2$};
\draw[thick,black] (0,1)--(1,1)--(2,0);
\end{tikzpicture}}
\caption{Set $S =\{(1,1),(2,0)\}$ is in $\Lambda(G_i)$ for $i \in \{1,2,3\}$ and $\Lambda(G_3) = \conv(S)$. Set $\conv(S)$ is shown in (a) and (b) with a different color. Observe that 1) any code whose service rate region contains the points in $S$ must contain the minimal region shown in (c), and 2) regions that are proper supersets of the minimal are achievable without enlarging the field size.} \label{fig:cool}
\end{figure*}

\begin{proposition}
\label{prop:reduced}
Let $S \subseteq \R^k$ and $G \in \F_q^{k \times n}$ be a full-rank matrix. We have 
\begin{enumerate}
    \item $S \subseteq \Lambda(G)$ implies $\conv(S) \subseteq \Lambda(G)$, where $\conv(S)$ is the convex hull of the points of $S$,
    \item $\conv(S) = \conv(S')$ where 
$$S' = \{s' \in S \mid s' \notin \conv(S \setminus \{s'\}\}.$$
\end{enumerate}
\end{proposition}
\begin{notation}
    In the sequel, $S \subseteq \R^k$ will always denote a nonempty set in reduced form. Following the notation of 
Proposition~\ref{prop:reduced},
we call $S'$
the \textbf{reduced form} of $S$. 

\end{notation}

In practice, we expect to know the demand vectors, and one needs to design systems (i.e., $G$) whose service rate region polytopes contain these demand vectors. Many systems can be chosen~(see, for instance, Figure \ref{fig:cool}), and 
the most efficient one will be selected (here ``efficiency" depends on the needs of the user). To this end, we study the following parameters, whose importance has been explained in Sec.~ \ref{sec:intro}.

\begin{definition}
\label{def:ourproblems}
When $q$ is fixed, we define
$$n_q(S) = \min\left\{n \in \Z_{\ge k} \mid \exists G \in \F_q^{k \times n} \mbox{ with } S \subseteq \Lambda(G)\right\}$$ and when $n$ is fixed, we define
$$q_n(S) = \min\left\{q \in \mP \mid \exists G \in \F_{q}^{k \times n} \mbox{ with } S \subseteq \Lambda(G) \right\}.$$
\end{definition}
The quantity $n_q(S)$ is a performance metric, and the minimum in its definition always exists since we can always use replication, resulting in many servers ($n$) compared to the optimal value $n_q(S)$.

\begin{example}
\label{ex:cool}
Let $S = \{(1,1),(2,0)\}$. We have $n_2(S)=3$. There are different polytopes $\Lambda(G)$ that contain $\conv(S)$ with~$G \in \F_2^{2 \times 3}$, see Figure \ref{fig:cool} for three of them.
\end{example}

While the number $n_q(S)$ always exists,  the existence of~$q_n(S)$ is not guaranteed.
\begin{example}
Let $S = \{(2,1),(1,2)\}$. We have $n_2(S) = 4$. The following argument shows that $q_3(S)$ does not exist. Regardless of the field size, one needs to fully use at least three servers for only one of the points in $S$. However, the servers used for the first point cannot support the remaining point. Therefore, there must be at least 4 servers.
\end{example}

\section{Methods and Tools}
\label{sec:methods}
We now propose a 
characterization of $n_q(S)$ that will allow us to apply mixed integer optimization to compute it.

\begin{remark}
\label{rem:expansion}
Let $q \in \mP$. Thus, $q=p^r$ for some prime~$p$ and positive integer $r$. Let $z$ denote a root of an irreducible polynomial in $\F_p[x]$ of degree $r$ in $\F_q$. We represent $\F_q$ as $\F_{p}^r$ using the function~$g : \F_q \rightarrow \F_p^r$, where $g(z^{j-1}) = e_{r+1-j}$ for all~$j \in \{1,\ldots,r\}$. For $h \in \F_q$, let $\overline{h} \in \{0,\ldots,q-1\}$ be the integer whose $p$-ary expansion is equal to $g(h)$. We then have 
$\Z_q = \{0,1,\ldots,q-1\}=\{\overline{h} \mid h \in \F_q\}.$

Let ${\bf{0}}_k$ be the all-zeros vector in $\Z_q^k$. We denote by $\overline{v_j} \in \Z_q^k \setminus {\bf{0}}_k$ the vector corresponding to the $q$-ary expansion of the integer $j \in \{1,\ldots,q^k-1\}$. Lastly, for all~$\overline{v_j} \in \Z_q^k \setminus {\bf{0}}_k$, we construct $v_j \in \F_q^k \setminus {\bf{0}}_k$ such that the $i$-th coordinate of $v_j$ is equal to $h \in \F_q$ when the $i$-th coordinate of $\overline{v_j}$ is $\overline{h} \in \Z_q$ for each $i \in \{1,\ldots,k\}.$ 

We also define following two sets:
\begin{align*}
D(H) &= \left\{d \in \left\{1,\ldots,\frac{q^k-1}{q-1}\right\} \ :\ u_d\in PG(k-1,q) \setminus H\right\}, \\
I(H) &= \{i \in \{1,\ldots,k\} \ :\ e_i \in PG(k-1,q) \setminus H\}.
\end{align*} 
for a hyperplane $H$ of $PG(k-1,q)$. In words, $D(H)$ and $I(H)$ are the set of projective points and standard basis vectors that do not lie on $H$, respectively.

\end{remark}
\noindent Using the representation of Remark \ref{rem:expansion} we get the following:

\begin{proposition}
\label{prop:nqS}
We have
\begin{multline*} 
    n_q(S) = \min\Biggl\{\sum_{j=1}^{q^k-1} n_j \mid \mbox{ there exists a matrix } G \\ \mbox{over $\F_q$ with $k$ rows, } |\{\nu \mid G^\nu = v_j\}| = n_j  \\ \mbox{ for all $j$, and } S \subseteq \Lambda(G) \Biggr\}.
\end{multline*}
\end{proposition} 

% Note that the matrices $G$ in Proposition~\ref{prop:nqS} 
% have $|\col(G)| = \sum_{i=1}^{q^k-1} n_i$.

%%%%%%%%%%%%%%%%%%%%%%%%%%%%%%%%%%%%%%%%%%%%%%%%%%%%%%%%%%%%%%%%%%%%%%%%%%%
%%%%%%%%%%%%%%%%%%%%%%%%%%%%%%%%%%%%%%%%%%%%%%%%%%%%%%%%%%%%%%%%%%%%%%%%%%%

We continue by translating the concept of recovery sets under the representation of Remark~\ref{rem:expansion}.

\begin{definition}
\label{def:minrec_updated}
Let $i \in \{1,\ldots,k\}$. A set $T \subseteq \{1,\ldots, q^k-1\}$ is called a \textbf{recovery set} for the $i$-th object if~$e_i \in \langle v_j \mid j \in T\rangle_{\F_q}$ and there is no $T' \subsetneq T$ with $e_i \in \langle v_j \mid j \in T'\rangle_{\F_q}$. We denote the set of recovery sets of the $i$-th object by $\mT_i$.
Note that~$|\mT_i| = |\mT_j|$ for all $i,j \in \{1,\ldots,k\}$.
\end{definition}

We illustrate the concepts introduced in Remark~\ref{rem:expansion} and Definition~\ref{def:minrec_updated} with an example.

\begin{example}
\label{ex:recovery}
Let $k=2$ and $q=3$. Then $v_1 = (0,1)$,~$v_2 = (0,2)$, $v_3 = (1,0)$, $v_4 = (1,1)$, $v_5 = (1,2)$, $v_6 = (2,0)$, $v_7 = (2,1)$, $v_8 = (2,2)$. Moreover, we have
\begin{multline*}
\mT_1 = \{\{1,8\},\{2,8\},\{3\},\{1,4\},\{2,4\},\{5,8\},\{1,5\},\\\{2,5\}, \{4,5\},\{6\},\{7,8\},\{1,7\},\{2,7\},\{4,7\}\},    
\end{multline*}
\begin{multline*}
\mT_2 =\{\{1\},\{2\},\{3,8\},\{3,4\},\{5,8\},\{3,5\},\{4,5\},\\\{6,8\},\{4,6\},\{5,6\},\{7,8\},\{3,7\},\{4,7\},\{6,7\}\}.    
\end{multline*}
\end{example}

The following result 
gives a sufficient condition for a demand vector to be supported under the terminology of Remark~\ref{rem:expansion}.
It can be applied to 
compute $n_q(S)$ for any set $S$ of finite cardinality using, for instance, mixed integer optimization; see~\cite{schrijver1998theory}.
Any feasible solution
of the program
gives an upper bound on $n_q(S)$.

\begin{theorem} 
\label{thm:system}
Assume that there exists a collection of non-negative integers $n_j$ for $j \in \{1,\ldots,q^k-1\}$ and a collection of non-negative real numbers $\{\theta_{i,T} \mid i \in \{1,\ldots,k\}, \, T \in \mT_i\}$
with the following properties:
\begin{empheq}[left = \empheqlbrace]{alignat=3}
        \sum_{T \in \mT_i}\theta_{i,T} &= \lambda_i &&\textnormal{ for } 1\leq i\leq k, \\
        \sum_{i=1}^k \sum_{\substack{T \in \mT_i \\ j \in T}} \theta_{i,T} &\leq n_j  \ &&\textnormal{ for } 1\leq j \leq q^k-1.
    \end{empheq}
There exists a system with $n = \sum_{j=1}^{q^k-1} n_j$ servers that supports (see Definition \ref{def:RandSRR}) the point $(\lambda_1,\ldots,\lambda_k) \in \R^k$.
\end{theorem}

The conditions in Theorem~\ref{thm:system} can be regarded as constraints of an optimization problem that minimizes the number of servers of a system such that the system supports the given list of demand vectors. We illustrate how to treat two demand vectors with an example focusing on a particularly tractable input set for convenience of exposition. This illustrates how Theorem~\ref{thm:system} is applied.

\begin{example}
\label{ex:S2}
Let $a,b \in \Z$ 
with $a \ge b > 0$, and~$S = \{(a,0),(0,b)\}$. We want to compute $n_2(S)$
using Theorem~\ref{thm:system} combined with optimization.
We have $v_1 = (0,1)$, $v_2 = (1,0)$, $v_3 = (1,1)$, $\mT_1 =\{\{1,3\},\{2\}\}$, and~$\mT_2 =\{\{1\},\{2,3\}\}$. By applying Theorem \ref{thm:system} to both demand vectors in $S$ we obtain that $n_2(S)$ equals the optimal value of the following integer linear optimization program (ILP). Here we slightly abuse the notation (instead of following the notation of Theorem \ref{thm:system}) to emphasize which variable of ILP contributes to which point of $S$: 
\begin{equation*}
\begin{array}{ll@{}ll}
\text{minimize}  & n_1 + n_2 + n_3 & &\\
\text{subject to}& 
s^1_{1,\{1,3\}} + s^1_{1,\{2\}} = a,  &  & \\
                 & s^2_{2,\{1\}} + s^2_{2,\{2,3\}} = b,  &  & \\
                 & s^1_{1,\{1,3\}}  \le n_1,\  s^1_{1,\{2\}} \le n_2,\ 
                 s^1_{1,\{1,3\}} \le n_3,&  & \\  
                 & s^2_{2,\{1\}} \le n_1,\  s^2_{2,\{2,3\}} \le n_2,\   s^2_{2,\{2,3\}} \le n_3,  &  & \\
                 & s^1_{1,\{1,3\}},\ s^1_{1,\{2\}},\ s^2_{2,\{1\}},\ s^2_{2,\{2,3\}} \in \R_{\ge 0},   &  & \\
                 & n_1,\ n_2,\ n_3 \in \Z_{\ge 0}.  &  & \\
\end{array}
\end{equation*}

We now point out a key observation on the role of the number of coded servers, namely $n_3$. We must have $n_3 \le n_1$ and $n_3 \le n_2$. This  follows from the fact $\{3\} \notin \mT_1 \cup \mT_2$, but~$\{1,3\} \in \mT_1$ and $\{2,3\} \in \mT_2$. Since $a \ge b$ by assumption, one can further assume that $n_2 \ge n_1 \ge n_3$. Thus, the problem reduces to: 
\begin{equation*}
\begin{array}{ll@{}ll}
\text{minimize}  & n_1 + n_2 + n_3 & &\\
\text{subject to}& n_2 + n_3 \ge a,  &  & \\
                 & n_1 + n_3 \ge b,  &  & \\
                 & n_2 - n_1 \ge 0,  &  & \\
                 & n_1 - n_3 \ge 0,  &  & \\
                 & n_1,\ n_2,\ n_3 \in \Z_{\ge 0}.   &  & \\
\end{array}
\end{equation*}

One feasible solution is$$(n_1,n_2,n_3) = (\ceil{b/2}, a - \floor{b/2},\floor{b/2}).$$ Therefore, we have $n_2(S) \le \ceil{a + b/2}$. 
We will later show that $n_2(S)\ge \ceil{a + b/2}$ by
Theorem \ref{thm:lb},
which in turn proves that 
$n_2(S) = a+\ceil{b/2}$. 
\end{example}

The following result
uses projective points 
and has been 
proven 
in \cite{kazemi2021efficient} for $q=2$ with 
slightly different notations.
We state it here
for arbitrary $q$. The proof is essentially the same.
We will apply later in Theorem \ref{thm:lb} to derive a lower bound on~$n_q(S)$ for some sets $S$ with a particular structure.

\begin{theorem}
\label{thm:geometry}
Let $S = \{s^1,\ldots,s^m\} \subseteq \R^k$. Let $G \in \F_q^{k \times n_q(S)}$ and let $H$ be a hyperplane of $\PG(k-1,q)$. Following the notation of Remark \ref{rem:expansion}, if $|\{\nu \mid G^\nu = v_d\}| = n_d$, then
$$\sum_{d \in D(H)} n_d \ge \Bigl\lceil \max \Bigl \{ \sum_{i \in I(H)} s^t_i \ : \ 1 \le t \le m \Bigr\}\Bigr\rceil.$$ 
\end{theorem}

%%%%%%%%%%%%%%%%%%%%%%%%%%%%%%%%%%%%%%%%%%%%%%%%%%%%%%%%%%%%%%%%%%%%%%%%%%%
%%%%%%%%%%%%%%%%%%%%%%%%%%%%%%%%%%%%%%%%%%%%%%%%%%%%%%%%%%%%%%%%%%%%%%%%%%%

\section{Upper and Lower Bounds}
\label{sec:bound}
This section focuses on bounds, existence results, and the interplay of the two parameters introduced in Definition~\ref{def:ourproblems}. 

We start with the simplest possible case, $|S|=1$. Here, we have to cover a single demand vector and determine the minimum number of servers and field size for which this is possible.

\begin{proposition}
\label{thm:s=1}
Let $s \in \R_{\ge 0}^k$ and $S = \{s\}$. For any $q$, we have $$n_q(S) = \sum_{i=1}^k \left\lceil s_i\right\rceil.$$ 
Furthermore, 
$$q_n(S) = \begin{cases}
        \textnormal{does not exist} & \textnormal{ if } n < \sum_{i=1}^k \left\lceil s_i\right\rceil,\\
        2 & \textnormal{ otherwise.}
    \end{cases}$$
\end{proposition}

The case where $|S|=2$ is significantly more involved then~$|S|=1$ and it will appear 
in the extended version of this work.
We continue with general upper and lower bounds on $n_q(S)$ for a set $S$ of any cardinality.

\begin{theorem}
\label{thm:sandwich}
Let $S = \{s^1,\ldots,s^m\} \subseteq \R^k$. We have $$\alpha(S) \le n_q(S) \le \beta(S),$$ where 
\begin{align*}
\alpha(S) &= \max\left\{\sum_{i=1}^k \left\lceil s^t_i\right\rceil \mid 1 \le t \le m \right\}, \\
\beta(S) &= \sum_{i=1}^k \max\left\{\left\lceil s^t_i\right\rceil \mid 1 \le t \le m \right\}.
\end{align*}
\end{theorem}

Obtaining bounds on $q_n(S)$ appears to be
more difficult, in general, than obtaining bounds on $n_q(S)$.
 The quantities introduced in Theorem \ref{thm:sandwich} can, however, be used to prove non-existence results on $q_n(S)$.

\begin{proposition}
\label{prop:notexist}
The quantity $q_{\alpha(S)}(S)$ never exists if $|S| > 1$ (and $S$ is in reduced form).
In particular, following the notation of Theorem~\ref{thm:sandwich}, if $|S| > 1$ then
$n_q(S) \ge \alpha(S)+1$ for any $q$.
\end{proposition}

In the remainder of this section, we focus on sets of the form~$S=\{X_ie_i \mid 1 \le i \le k\}$, for a given vector~$X \in \R^k_{\ge 0}$.
Note that the convex hull of such a set $S$ is a simplex polytope in the positive orthant, one of the best-known
polytopes.

The next result establishes a lower bound for the quantity~$n_q(S)$, which explicitly depends on the underlying field size~$q$. Its proof uses Theorem~\ref{thm:geometry}, which we stated as a preliminary result in Sec.~\ref{sec:methods}.

\begin{theorem}
\label{thm:lb}
Let $(X_i)_{i=1}^k$ be a non-increasing sequence of non-negative integers and let $S=\{X_ie_i \mid 1 \le i \le k\}.$
For all $q$, we have 
\begin{equation*}
\label{lowerbound}
n_q(S) \ge  \left \lceil \sum\limits_{i=1}^k  
q^{1-i}X_i  \right \rceil.
\end{equation*} 
\end{theorem}
\begin{proof}
Let $\smash{\qbin{a}{b}{q}}$ denote the number of $b$-dimensional subspaces of an $a$-dimensional vector space over $\F_q$ for non-negative integers $a \ge b$; see \cite{andrews1998theory}.
We apply Theorem~\ref{thm:geometry} to all hyperplanes of $\PG(k-1,q)$. This gives $\smash{\qbin{k}{k-1}{q} = \qbin{k}{1}{q}}$ inequalities, since $\PG(k-1,q)$ has that many hyperplanes. We can represent these inequalities in a system 
\begin{equation} \label{ineqs}
A\bf{x} \ge \bf{b}
\end{equation}
for a matrix $A$ of suitable size and vectors $\bf{x}$ and $\bf{b}$ of length $\smash{\qbin{k}{1}{q}}$. Since $\smash{(X_i)_{i=1}^k}$ is a non-increasing sequence of integers, we have 
$$|\{d \mid b_d = X_i\}| = \qbin{k-i+1}{1}{q} - \qbin{k-i}{1}{q}$$ for all $1 \le i \le k$. Moreover, following the notation of Theorem~\ref{thm:geometry}, each $n_d$ occurs in exactly $|D(H)|$ inequalities.

Let $Z = \qbin{k-i+1}{1}{q} - \qbin{k-i}{1}{q}$ for the rest of the proof. Summing all the inequalities in~\eqref{ineqs} and dividing by~$|D(H)|$ gives
\allowdisplaybreaks
\begin{align*}
n_q(S) &\ge \frac{1}{|D(H)|}\sum_{i=1}^k Z\ceil{X_i} = \frac{1}{|D(H)|}\sum_{i=1}^k ZX_i \\
&= \frac{1}{\qbin{k}{1}{q} - \qbin{k-1}{1}{q}}\sum_{i=1}^k ZX_i  \\
&= X_1 + \frac{q-1}{q^k-q^{k-1}} \sum_{i=2}^k ZX_i \\
&= X_1 + \frac{X_2}{q} + \ldots + \frac{X_k}{q^{k-1}}.
\end{align*}
Taking the ceiling, we bound on the integer $n_q(S)$.
\end{proof}

When $q=2$, we can establish an upper bound for~$n_q(S)$, where $S$ has the same properties stated in Theorem~\ref{thm:lb}. The proof of the following result is constructive and will appear in the extended version of this work because of space constraints.

\begin{theorem} \label{thm:lb2}
Let $(X_i)_{i=1}^k$ be a non-increasing sequence of non-negative integers and $S=\{X_ie_i \mid 1 \le i \le k\}$. We have 
\begin{equation*}
n_2(S) \le  \sum\limits_{i=1}^k \left \lceil \frac{X_i}{2^{i-1}}  \right \rceil.
\end{equation*}
\end{theorem}

By combining Theorems~\ref{thm:lb} and~\ref{thm:lb2} we obtain the following result for the case $q=2$.

\begin{corollary}
\label{cor:simplex}
Let $(X_i)_{i=1}^k$ be a non-increasing sequence of non-negative integers and let $S=\{X_ie_i \mid 1 \le i \le k\}$. Then
\begin{equation*}
\left \lceil \sum\limits_{i=1}^k  
\frac{X_i}{2^{i-1}}  \right \rceil\le n_2(S) \le  \sum\limits_{i=1}^k \left \lceil \frac{X_i}{2^{i-1}}  \right \rceil.
\end{equation*}
\end{corollary}

Note that the lower and upper bounds of Corollary \ref{cor:simplex} coincide when $2^{i-1}$ divides $X_i$ for all $i \in \{2,\ldots,k\}.$ Moreover, if~$X_i = 2^{k-1}$ for all $i \in \{1,\ldots,k\}$, then the well-known binary simplex code (of length $2^k-1$) attains $n_2(S)$.

%%%%%%%%%%%%%%%%%%%%%%%%%%%%%%%%%%%%%%%%%%%%%%%%%%%%%%%%%%%%%%%%%%%%%%%%%%%
%%%%%%%%%%%%%%%%%%%%%%%%%%%%%%%%%%%%%%%%%%%%%%%%%%%%%%%%%%%%%%%%%%%%%%%%%%%

\section{Applications and Examples}
\label{sec:interplay}

We provide evidence of how the established 
results can be applied in concrete scenarios. The proofs of the results in this paper use different techniques, which capture different mathematical properties of the quantities~$n_q(S)$ and~$q_n(S)$. As a natural consequence, the results perform best when combined. For instance,
Theorem \ref{thm:sandwich} together with Proposition \ref{prop:notexist} can be used to compute the value of~$n_q(S)$ in some cases.

\begin{example}
Let $S = \{(2.5,1),(1,2)\}$. Following the notation of Theorem~\ref{thm:sandwich}, we have $\alpha(S) = 4$ and $\beta(S)=5$. Therefore,~$4 \le n_q(S) \le 5$. The number $q_4(S)$ does not exist for any $q\in \mP$ by Proposition \ref{prop:notexist}. Thus $n_q(S)$ has to be equal to~$5$ for every $q \in \mP$.
\end{example}

A general observation is that
bounds on $n_q(S)$ that depend on $q$
automatically give information on $q_n(S)$. 
That is the case of 
Theorem~\ref{thm:lb}, which implies the following result.

\begin{corollary}
\label{cor:qnSfromnqS}
Let $(X_i)_{i=1}^k$ be a non-increasing sequence of non-negative integers and let
$S=\{X_ie_i \mid 1 \le i \le k\}$. For any $n \in \Z_{\ge k}$ s.t.\ $q_n(S)$ exists, we have 
$$q_n(S) \ge \frac{X_2}{n-X_1}.$$
\end{corollary}
\begin{proof}

Let $n \in \Z_{\ge k}$ s.t.\ $q_n(S)$ exists. By Theorem \ref{thm:lb}, we have $$n \ge \sum\limits_{i=1}^k \left \lceil q^{1-i}X_i  \right \rceil \ge X_1 + \frac{X_2}{q},$$ yielding the result.
\end{proof}

Even though 
Corollary \ref{cor:qnSfromnqS} is a quite coarse approximation of Theorem~\ref{thm:lb},
it provides very good bounds in some instances.

\begin{example}
Let $S = \{(100,0),(0,99)\}$. Corollary \ref{cor:qnSfromnqS} gives~$q_{100}(S) \ge 99$ which implies that $q_{100}(S) \ge 101$ since it must be a prime power. It can also be checked that~$n_2(S) =150$ in this example.
\end{example}

We conclude the paper by showing that Corollary \ref{cor:qnSfromnqS} can actually be met with equality for some parameter sets.

\begin{example}
Let $S = \{(2,0),(0,2)\}$ and $n=3$. Then Corollary \ref{cor:qnSfromnqS} implies that $q_3(S) \ge 2$ and it is met with equality, since $S \subseteq \Lambda(G)$ for $$G=
\begin{pmatrix}
1 & 0 & 1   \\
0 & 1 & 1  
\end{pmatrix}\in\F_2^{2\times 3}.$$
\end{example}

\section*{Conclusions and Future Work}
We addressed the problem of computing the minimum number of servers and minimum alphabet size to support a demand vector for the service rate region problem. We proposed a method to obtain the exact values of the first parameter
based on optimization.
We then established upper and lower bounds for both parameters, studied their interplay, and discussed their sharpness. Future work includes sharpening the bounds and investigating more general sets of demand vectors than those considered in this paper.

%%%%%%%%%%%%%%%%%%%%%%%%%%%%%%%%%%%%%%%%%%%%%%%%%%%%%%%%%%%%%%%%%%%%%%%%%%%
\newpage
%%%%%%%%%%%%%%%%%%%%%%%%%%%%%%%%%%%%%%%%%%%%%%%%%%%%%%%%%%%%%%%%%%%%%%%%%%%

\bibliography{isit}
\bibliographystyle{ieeetr}

\end{document}